\documentclass{article}
\usepackage{amsfonts}
\usepackage{amsmath}
\usepackage{amscd}

\newtheorem{theorem}{Theorem}

\newtheorem{corollary}[theorem]{Corollary}

\newtheorem{definition}[theorem]{Definition}

\newtheorem{lemma}[theorem]{Lemma}

\newtheorem{proposition}[theorem]{Proposition}

\newenvironment{proof}[1][Proof]{\noindent\textbf{#1.} }{\ \rule{0.5em}{0.5em}}
\input{tcilatex}
\begin{document}

\title{Energy-momentum tensors in classical field theories - a modern
perspective}
\author{Nicoleta Voicu \\
"Transilvania"\ University, Brasov, Romania}
\date{}
\maketitle

\begin{abstract}
The paper presents a general geometric approach to energy-momentum tensors
in Lagrangian field theories, based on a Hilbert-type definition.

The approach is consistent with the ones defining energy-momentum tensors in
terms of hypermomentum maps given by the diffeomorphism invariance of the
Lagrangian \ - and, in a sense, complementary to these, with the advantage
of an increased simplicity of proofs and also, opening up new insights into
the topic. A special attention is paid to the particular cases of metric and
metric-affine theories.
\end{abstract}


\section{Introduction}

The paper presents a geometric approach to energy-momentum tensors in
general Lagrangian field theories. Rather than trying to offer a detailed
review of the long and intricate history of the topic, it aims to bring more
simplicity and clarity - by looking at the problem from a somewhat different
perspective.

This perspective is nothing but an extension to arbitrary Lagrangian field
theories of the general relativistic one, based on a Hilbert-type definition
of the energy-momentum tensor.

The energy-momentum tensor introduced this way is, of course, not new - and
neither are its usual properties (generalized covariant conservation law,
gauge invariance etc.). When evaluated on critical sections of the matter
Lagrangian, the Hilbert-type energy-momentum tensor agrees (up to a sign)
with the "improved Noether current", as introduced by Gotay and Marsden, 
\cite{Gotay}. Still, there are some advantages of using a Hilbert-type
approach over a Noether-type one. The first one is simplicity, both in the
calculation of the energy-momentum tensor and in the proofs of its
properties.

The Hilbert-type formula also opens up the possibility of using results of
the inverse problem of variational calculus, such as:\ the classification of
first-order energy-momentum source forms, \cite{Krupka-em-tensors}, or the
notion of variational completion, \cite{inv-problem-book}. In the situation
when we know, e.g., by some empirical method, a term of an energy-momentum
tensor, the latter allows one to recover its full expression, together with
the corresponding Lagrangian.

\bigskip

Before passing to the technical details of our story, let us present in
brief the main problems surrounding energy-momentum tensors.

- The \textit{canonical, or Noether\ energy-momentum tensor }is\textit{\ }%
defined, in special relativity, by means of the Noether current due to the
invariance of the Lagrangian to the group of spacetime \textit{translations}%
. It is conserved on-shell and its time-time and time-space components give
the correct energy and momentum densities of the system. Still, as it is
generally neither symmetric, nor gauge-invariant, it requires an
"improvement"\ procedure; the classical special-relativistic recipe
(Belinfante\&Rosenfeld, 1940) is based on enlarging the considered symmetry
group to the whole Poincar\'{e} group.

-\ General relativity came with a completely different toolkit. The \textit{%
Hilbert, or metric energy-momentum tensor,} given by the variational
derivative of the matter Lagrangian with respect to the metric, is
symmetric, gauge-invariant and, as a result of the diffeomorphism invariance
of the Lagrangian, its covariant divergence vanishes on--shell. Hence, it
does not require any improvement procedure. But, on the other hand, \cite%
{Gotay}, it is not obvious at all that it gives the correct energy and
momentum densities of the system under discussion.

\bigskip

It thus appeared the idea of obtaining the energy-momentum tensor of
basically \textit{any} classical field theory as a kind of improved Noether
current, which should coincide, in the case of general relativity, with the
Hilbert one. But this proved to be a highly non-trivial task - and gave rise
to long-standing debates (a good account of which is given, e.g., in \cite%
{Gotay}).

The first problem which appears is the one of the symmetry group to be
considered. On a general manifold $X,$ translations - let alone the Poincar%
\'{e} group -\ make no sense; the natural choice in this case seems to be
the group $Diff(X)$ of diffeomorphisms of $X.$ But, as $Diff(X)$ is
infinite-dimensional, no non-trivial Noether current can be obtained from
it; the corresponding "Noether current" (called a \textit{hypermomentum map}%
) is always zero when \textit{all} the variables are subject to
Euler-Lagrange equations, \cite{Forger}, \cite{Gotay}. The way out of this
impasse consists, \cite{Forger}, \cite{Giachetta1}, in dividing the
variables of the theory into \textit{background }ones (e.g., a metric and/or
a connection, a tetrad etc.) and \textit{dynamical }(or \textit{matter})\
ones, and, accordingly, in splitting the total Lagrangian $\lambda $ into a
sum%
\begin{equation*}
\lambda =\lambda _{b}+\lambda _{m},
\end{equation*}%
where the \textit{background Lagrangian }$\lambda _{b}$ only depends on the
background variables and their derivatives up to some order, while the 
\textit{matter Lagrangian} $\lambda _{m}$ may depend on all the variables of
the theory. For the matter Lagrangian $\lambda _{m},$ the background
variables will no longer be subject to any Euler-Lagrange equations (these
are only supposed to obey Euler-Lagrange equations for the \textit{total}
Lagrangian $\lambda $). This way, one can obtain a nonzero hypermomentum
map, which solves the problem.

\bigskip

Improving Noether currents is the path taken by the majority of the authors.
A general, geometric and systematic approach in this respect is the one by
Gotay and Marsden (\cite{Gotay}, 2001) - with some refinements brought by
Forger and R\H{o}mer\footnote{%
The main advance brought in \cite{Forger} resides in the fact that\textit{\ }%
the energy-momentum tensor is built from the matter Lagrangian $\lambda _{m}$
only (while in \cite{Gotay}, it is built from the total Lagrangian $\lambda $
- and thus vanishes on-shell). Also, the energy-momentum tensor is regarded
as a geometric object on a jet bundle of the total configuration manifold,
rather than on the spacetime manifold $X$ - a standpoint which we will also
adopt in the following.}, \cite{Forger}. The approach is extended to higher
order variational problems of differential index 1 by Fern\'{a}ndez, Garc%
\'{\i}a and Rodrigo, \cite{Fernandez}.

Roughly speaking, the main result in \cite{Gotay} states that, if a
first-order Lagrangian on a fibered manifold $Y\overset{\pi }{\rightarrow }X$
is invariant to the flow of an arbitrary vector field\footnote{%
As diffeomorphisms of $X$ act on $X,$ not on $Y,$ an embedding $%
Diff(X)\rightarrow Aut(Y)\ $is needed in order to correctly define this
invariance. A\ rigorous definition will be given below.} $\xi $ on $X$ and $%
\mathcal{J}_{\xi }$ is the corresponding Noether current, then, for any
solution $\gamma :X\rightarrow Y$ of the Euler-Lagrange equations, there
uniquely exists a (1,1)-tensor density $\mathcal{T}(\gamma )=\mathcal{T}%
_{~j}^{i}(\gamma )\dfrac{\partial }{\partial x^{i}}\otimes dx^{j}$ on $X$
such that, for all compact hypersurfaces $\Sigma \subset X,$ 
\begin{equation}
\underset{\Sigma }{\int }J^{1}\gamma ^{\ast }\mathcal{J}_{\xi }=\underset{%
\Sigma }{\int }\mathcal{T}_{~j}^{i}(\gamma )\xi ^{j}\omega _{i}
\label{ultralocality}
\end{equation}%
(here, $\omega _{i}=dx^{0}\wedge ...\wedge dx^{i-1}\wedge dx^{i+1}\wedge
...\wedge dx^{n}$).

The improved canonical energy-momentum tensor density (\ref{ultralocality})
is gauge invariant and given (up to a sign) by a Hilbert-type formula.
Moreover,\ relation (\ref{ultralocality}) ensures that $\mathcal{T}(\gamma )$
is a "physically correct" energy-momentum tensor, in the following sense, 
\cite{Gotay}. If $\Sigma $ is a Cauchy hypersurface and the vector field $%
\xi $ is transversal to $\Sigma ,$ then the Hamiltonian (energy)\
corresponding to the "direction of evolution" $\xi ,$ is $H_{\xi }=-\underset%
{\Sigma }{\int }J^{1}\gamma ^{\ast }\mathcal{J}_{\xi },$ i.e., it is given,
again up to a sign, by (\ref{ultralocality}).

\bigskip

The latter remark points out that one could hardly overestimate the
importance of having the energy-momentum tensor related to Noether currents
as in (\ref{ultralocality}). But there is another way of looking at the same
relation. Instead of using it as a definition, we will prefer to obtain it
as a consequence of a Hilbert-type definition.

We will assume, as in \cite{Gotay}, \cite{Fernandez}, that the differential
index of the theory in the background variables is 1. Under this assumption,
in the Euler-Lagrange expression of the matter Lagrangian $\lambda _{m}$
with respect to the background variables, we can isolate a total divergence
term - which leads to the correct energy-momentum tensor. The remaining
terms give a generalized covariant conservation law (an \textit{%
energy-momentum balance law}) - which is obtained as an immediate
consequence of the first variation formula.

These results hold true regardless of the order of the Lagrangian or of the
nature of the coupling (minimal or non-minimal)\ between the background and
the matter variables.

\bigskip

The paper is structured as follows. In Section 2, we present a quite
detailed overview of the necessary notions and results. Section 3 is devoted
to the definition of the energy-momentum tensor and Section 4, to its main
properties, namely, energy-momentum balance law and gauge invariance.
Section 5 discusses two major particular cases: metric and metric-affine
backgrounds. In the metric-affine case, the obtained energy-momentum balance
law has a simpler (and explicitly covariant) expression compared to the
known ones, \cite{Obukhov}, \cite{Lompay}. The last section presents an
application to energy-momentum tensors of the notion of variational
completion.

\section{Variational calculus - the modern framework}

The language of differential forms, in which a Lagrangian is regarded as a
differential form on a certain jet bundle, allows a concise, coordinate-free
formulation of variational calculus on arbitrary manifolds. In this
approach, the variation of a Lagrangian $\lambda $ is understood as its Lie
derivative with respect to a vector field on the respective jet bundle. The
notions and results presented in this section can be found in more detail,
e.g., in the monograph \cite{Krupka-book}.

\subsection{Differential forms on jet prolongations of a fibered manifold 
\label{Definitions}}

Consider a fibered manifold $Y$ of dimension $m+n,$ with connected,
orientable $n$-dimensional base $X$ and projection $\pi :Y\rightarrow X.$
The manifold $Y$ will be called the \textit{configuration manifold} and $X,$
the \textit{spacetime manifold}. Fibered charts $(V,\psi )$, $\psi
=(x^{i},y^{\sigma })$ on $Y$ induce the fibered charts $(V^{r},\psi ^{r}),$ $%
\psi ^{r}=(x^{i},y^{\sigma },y_{~j_{1}}^{\sigma
},...,y_{~j_{1}j_{2}...j_{r}}^{\sigma })$ on the $r$-jet prolongation $%
J^{r}Y $ of $Y$ and $(U,\phi ),$ $\phi =(x^{i})$ on $X.$ We denote by $\pi
^{r,s}$ the canonical projections $\pi ^{r,s}:J^{r}Y\rightarrow J^{s}Y,~\ \
(x^{i},y^{\sigma },y_{~j_{1}}^{\sigma },...,y_{~j_{1}j_{2}...j_{r}}^{\sigma
})\mapsto (x^{i},y^{\sigma },y_{~j_{1}}^{\sigma
},...,y_{~j_{1}j_{2}...j_{s}}^{\sigma })$ ($r>s,$ $J^{0}Y:=Y$) and by $\pi
^{r},$ the projection $\pi ^{r}:J^{r}Y\rightarrow X,~(x^{i},y^{\sigma
},y_{~j_{1}}^{\sigma },...,y_{~j_{1}j_{2}...j_{r}}^{\sigma })\mapsto
(x^{i}). $ By $\Omega _{k}^{r}Y,$ we will mean the set of $k$-forms of order 
$r$ over $Y$. In particular, $\mathcal{F}(Y):=\Omega _{0}^{r}Y$ is the set
of real-valued smooth functions over $J^{r}Y.$ The set of $\mathcal{C}%
^{\infty } $-smooth sections of $Y$ will be denoted by $\Gamma (Y);$ its
elements $\gamma :X\rightarrow Y,$ $(x^{i})\mapsto (y^{\sigma }(x^{i}))$
will be called \textit{fields. }The notation $\mathcal{X}(Y)$ will mean the
module of vector fields on $Y.$

\bigskip

\textbf{Horizontal forms. Lagrangians. }A differential form $\rho \in \Omega
_{k}^{r}Y$ is called $\pi ^{r}$-horizontal, or, simply, \textit{horizontal},
if $\rho (\Xi _{1},...,\Xi _{k})=0$ whenever one of the vector fields $\Xi
_{i},$ $i=\overline{1,k},$ is $\pi ^{r}$-vertical (i.e., $T\pi ^{r}(\Xi
_{i})=0$). Any horizontal form on $\Omega _{k}^{r}Y$ is locally expressed as:%
\begin{equation}
\rho =\dfrac{1}{k!}A_{i_{1}i_{2}...i_{k}}dx^{i_{1}}\wedge dx^{i_{2}}\wedge
...\wedge dx^{i_{k}},  \label{horizontal_form}
\end{equation}%
where the coefficients $A_{i_{1}i_{2}...i_{k}},$ $k\leq n,$ are smooth
functions of the coordinates $x^{i},y^{\sigma },y_{~j_{1}}^{\sigma
},...,y_{~j_{1}j_{2}...j_{r}}^{\sigma }$. Similarly, one can introduce $\pi
^{r,s}$\textit{-horizontal} forms, $0\leq s\leq r$; locally, these are
generated by exterior products of the differentials $dx^{i},dy^{\sigma
},...,dy_{~j_{1}...j_{s}}^{\sigma }.$

In particular, a horizontal $n$-form $\lambda \in \Omega _{n}^{r}Y,$ where $%
n=\dim X,$ is called a \textit{Lagrangian. }Locally, a Lagrangian is
expressed as:%
\begin{equation}
\lambda =\mathcal{L}\omega _{0},~\ \ \ \ \ \ \mathcal{L=L}(x^{i},y^{\sigma
},...,y_{i_{1}...i_{r}}^{\sigma }),
\end{equation}%
where $\omega _{0}:=dx^{1}\wedge ...\wedge dx^{n}.$

\bigskip

Any differential $k$-form $\rho \in \Omega _{k}^{r}(Y)$ can be transformed
into a horizontal $k$-form of order $r+1.$ This is achieved by means of the 
\textit{horizontalization }operator, which is the (unique)\ morphism of
exterior algebras $h:\Omega ^{r}(Y)\rightarrow \Omega ^{r+1}(Y)$ (where $%
\Omega ^{s}Y$ means the set of all differential forms of order $s$ over $Y$)
such that, for any $f\in \mathcal{F}(Y)$ and any fibered chart: $hf=f\circ
\pi ^{r+1,r}~$and$\ hdf=d_{i}fdx^{i}.$ Here,%
\begin{equation*}
d_{i}f:=\partial _{i}f+\dfrac{\partial f}{\partial y^{\sigma }}%
y_{~i}^{\sigma }+...\dfrac{\partial f}{\partial y_{~j_{1}...j_{r}}^{\sigma }}%
y_{~j_{1}...j_{r}i}^{\sigma }.
\end{equation*}%
On the natural basis 1-forms, it acts as: 
\begin{equation}
hdx^{i}:=dx^{i},~\ hdy^{\sigma }=y_{~i}^{\sigma
}dx^{i},...,hdy_{~j_{1}...j_{k}}^{\sigma }=y_{~j_{1}...j_{k}i}^{\sigma
}dx^{i},\ ~\ \ k=\overline{1,r}.  \label{horizontalization_basis}
\end{equation}%
Moreover, for any $\rho \in \Omega _{k}^{r}(Y),$ there holds the equality:%
\begin{equation}
J^{r}\gamma ^{\ast }\rho =J^{r+1}\gamma ^{\ast }(h\rho ),~\ \ \forall \gamma
\in \Gamma (Y).  \label{property_h}
\end{equation}

If two $\pi ^{r}$-horizontal forms $\rho ,\theta \in \Omega _{k}^{r}(Y)$
satisfy $J^{r}\gamma ^{\ast }\rho =J^{r}\gamma ^{\ast }\theta $ for any
section $\gamma \in \Gamma (Y),$ then $\rho =\theta .$

\bigskip

\textbf{Contact forms and first canonical decomposition of a differential
form on }$J^{r}Y$\textbf{. }A form $\theta \in \Omega _{k}^{r}Y$ is a 
\textit{contact form }of order $r$ on $Y$ if it is annihilated by all jets $%
J^{r}\gamma $ of sections $\gamma \in \Gamma (Y);$ equivalently, $\theta $
is a contact form if and only if it belongs to the kernel of $h.$

For instance, 
\begin{eqnarray}
&&\omega ^{\sigma }=dy^{\sigma }-y_{~j}^{\sigma }dx^{j},~\ \ \omega
_{~i_{1}}^{\sigma }=dy_{~i_{1}}^{\sigma }-y_{~i_{1}j}^{\sigma }dx^{j},...
\label{contact_1-forms} \\
&&~\ \ \ \omega _{~i_{1}i_{2}...i_{r-1}}^{\sigma
}=dy_{~i_{1}i_{2}...i_{r-1}}^{\sigma }-y_{~i_{1}i_{2}...i_{r-1}j}^{\sigma
}dx^{j},  \notag
\end{eqnarray}%
represent contact forms on $J^{r}Y.$ These contact forms can be used in
order to construct a local basis of the module of 1-forms over $J^{r}Y,$
called the \textit{contact basis: }$\{dx^{i},\omega ^{\sigma },....,\omega
_{~i_{1}...i_{r-1}}^{\sigma },dy_{~i_{1}...i_{r}}^{\sigma }\}.$

A $k$-form $\theta \in \Omega _{k}^{r}Y$ is $l$\textit{-contact (}$l\leq k$)%
\textit{\ }if, corresponding to any fibered chart on $Y$, in the
decomposition of the pulled back form $(\pi ^{r+1,r})^{\ast }\theta \in
\Omega _{k}^{r+1}Y$ in the contact basis over $J^{r+1}Y,$ each term contains
exactly $l$ contact 1-forms (\ref{contact_1-forms}), i.e., 
\begin{equation*}
(\pi ^{r+1,r})^{\ast }\theta =\dfrac{1}{l!(k-l)!}\theta
_{A_{1}...A_{l},i_{l+1},...,i_{k}}\omega ^{A_{1}}\wedge ...\wedge \omega
^{A_{l}}\wedge dx^{i_{l+1}}\wedge ...\wedge dx^{i_{k}},
\end{equation*}%
where $A\in \{\sigma ,(_{j}^{\sigma }),...,(_{j_{1}...j_{r}}^{\sigma })\}$.

An arbitrary $k$-form $\rho \in \Omega _{k}^{r}Y$ admits a unique
decomposition (called the \textit{first canonical decomposition}):%
\begin{equation}
(\pi ^{r+1,r})^{\ast }\rho =h\rho +p_{1}\rho +...+p_{k}\rho ,
\label{first_canonical_decomp}
\end{equation}%
where the form $p_{l}\rho $ is $l$-contact, $l=1,...,k$. The sum $p\rho
=p_{1}\rho +...+p_{k}\rho $ is the \textit{contact component} of $\rho .$

A $\pi ^{r,0}$-horizontal, 1-contact $(n+1)$-form $\eta \in \Omega
_{n+1}^{r}Y$ is called a \textit{source form}. Locally, a source form is
expressed as: 
\begin{equation}
\eta =\eta _{\sigma }\omega ^{\sigma }\wedge \omega _{0}.
\label{source_form}
\end{equation}

\bigskip

With respect to coordinate changes on $J^{r}Y$ induced by fibered coordinate
changes $x^{i}=x^{i}(x^{i^{\prime }}),y^{\sigma }=y^{\sigma }(x^{i^{\prime
}},y^{\sigma ^{\prime }}),$ we have:%
\begin{eqnarray}
dx^{i} &=&\dfrac{\partial x^{i}}{\partial x^{i^{\prime }}}dx^{i^{\prime
}},~\ \omega ^{\sigma }=\dfrac{\partial y^{\sigma }}{\partial y^{\sigma
^{\prime }}}\omega ^{\sigma ^{\prime }},  \label{transf_contact_forms} \\
\omega _{0} &=&\det (\dfrac{\partial x^{j}}{\partial x^{j^{\prime }}})\omega
_{0^{\prime }},\ \omega _{i}=\dfrac{\partial x^{i^{\prime }}}{\partial x^{i}}%
\det (\dfrac{\partial x^{j}}{\partial x^{j^{\prime }}})\omega _{i^{\prime }},
\label{transf_omega0}
\end{eqnarray}%
where $\omega _{i}=\mathbf{i}_{\partial /\partial x^{i}}\omega
_{0}=(-1)^{i-1}dx^{1}\wedge ....dx^{i-1}\wedge dx^{i+1}\wedge ...\wedge
dx^{n}$ and 
\begin{equation}
\eta _{\sigma ^{\prime }}=\dfrac{\partial y^{\sigma }}{\partial y^{\sigma
^{\prime }}}\det (\dfrac{\partial x^{i}}{\partial x^{i^{\prime }}})\eta
_{\sigma }  \label{transf_source_form}
\end{equation}%
for the components of a (globally defined) source form (\ref{source_form})
on $Y$.

\subsection{Lepage equivalents of a Lagrangian and first variation formula 
\label{Lepage_equiv}}

\textbf{Lepage equivalents. }Consider a Lagrangian%
\begin{equation}
\lambda =\mathcal{L}\omega _{0}\in \Omega _{n}^{r}Y,~\ \ \ \ \ \ \mathcal{L=L%
}(x^{i},y^{\sigma },...,y_{i_{1}...i_{r}}^{\sigma }).
\label{general_Lagrangian}
\end{equation}%
The \textit{action\ }attached to the Lagrangian (\ref{general_Lagrangian})
and to a compact domain $D\subset X\ $is the function $S:\Gamma
(Y)\rightarrow \mathbb{R},$ given by: 
\begin{equation*}
S(\gamma )=\underset{D}{\int }J^{r}\gamma ^{\ast }\lambda .
\end{equation*}%
The \textit{variation} $\delta _{\Xi }S:\Gamma (Y)\rightarrow \mathbb{R}$ of 
$S$ under the flow of a vector field $\Xi =\xi ^{i}\partial _{i}+\Xi
^{\sigma }\partial _{\sigma }$ on $Y$ is given, \cite{Krupka-book}, by the
Lie derivative $\partial _{J^{r}\Xi }\lambda $ of $\lambda $ with respect to
the prolongation $J^{r}\Xi $:%
\begin{equation}
\delta _{\Xi }S(\gamma )=\underset{D}{\int }J^{r}\gamma ^{\ast }\partial
_{J^{r}\Xi }\lambda .  \label{variation}
\end{equation}%
A section $\gamma :X\rightarrow Y,(x^{i})\mapsto y^{\sigma }(x^{i})$ is
called a \textit{critical section\ }for $S,$ if the variation $\delta _{\Xi
}S(\gamma )$ vanishes, i.e., if:%
\begin{equation}
\underset{D}{\int }J^{r}\gamma ^{\ast }(\partial _{J^{r}\Xi }\lambda )=0.
\label{pulled_back_variation}
\end{equation}

\bigskip

A\textit{\ Lepage equivalent}\ of a Lagrangian $\lambda \in \Omega _{n}^{r}Y$
is an $n$-form $\theta _{\lambda }$ on some jet prolongation $J^{s}Y,$ with
the following properties:

(1) $\theta _{\lambda }$ defines the same variational problem as $\lambda ,$
i.e.:%
\begin{equation}
(\pi ^{q,s+1})^{\ast }h\theta _{\lambda }=(\pi ^{q,r})^{\ast }\lambda ,
\label{Lepage_def1}
\end{equation}%
where $q:=\max (s+1,r);$

(2)\ the first contact component $p_{1}d\theta _{\lambda }$ is a source form
(i.e., it is generated by $\omega ^{\sigma }$ alone).

\bigskip

Every Lagrangian $\lambda \in \Omega _{n}^{r}Y$ admits Lepage equivalents $%
\theta _{\lambda }$. Corresponding to any local chart, these Lepage
equivalents are expressed, \cite{Krupka-book}, as:\ $\theta _{\lambda
}=\Theta _{\lambda }+d\eta +\mu ,$ where $\eta $ is a contact form, $\mu $
is at least 2-contact and%
\begin{equation}
\Theta _{\lambda }=\mathcal{L}\omega _{0}+(\overset{r-1}{\underset{k=0}{\sum 
}}\overset{r-1-k}{\underset{l=0}{\sum }}(-1)^{l}d_{p_{1}}...d_{p_{l}}\dfrac{%
\partial \mathcal{L}}{\partial y_{~j_{1}...j_{k}p_{1}...p_{l}i}^{\sigma }}~\
\ \omega _{j_{1}...j_{k}}^{\sigma }\wedge \omega _{i};
\label{Poincare-Cartan_general}
\end{equation}%
the order of $\Theta _{\lambda }$ is $s\leq 2r-1$.

The notion of Lepage equivalent generalizes the idea of Poincar\'{e}-Cartan
form from mechanics and allows one to obtain a coordinate-free description
of both Euler-Lagrange equations and Noether theorem, as follows.

\bigskip

\textbf{First variation formula. }Let $\lambda \in \Omega _{n}^{r}Y$ be a
Lagrangian, $\theta _{\lambda }\in \Omega _{n}^{s}Y,$ an arbitrary Lepage
equivalent of $\lambda ,$ of some order $s$ and $\Xi \in \mathcal{X}(Y),$ a
projectable vector field. Then, there holds the \textit{first variation
formula, }\cite{Krupka-book},%
\begin{equation}
J^{r}\gamma ^{\ast }(\partial _{J^{r}\Xi }\lambda )=J^{s+1}\gamma ^{\ast }%
\mathbf{i}_{J^{s+1}\Xi }(p_{1}d\theta _{\lambda })+d(J^{s}\gamma ^{\ast }%
\mathbf{i}_{J^{s}\Xi }\theta _{\lambda }).  \label{first_variation_refined}
\end{equation}%
The terms in the right hand side of (\ref{first_variation_integral}) have
the following meaning:

\textit{i) }the source form $p_{1}d\theta _{\lambda }$ is the \textit{%
Euler-Lagrange form} attached to $\lambda $:%
\begin{equation}
p_{1}d\theta _{\lambda }=E(\lambda )\in \Omega _{n+1}^{s+1}(Y),
\label{EL_form_coord_free}
\end{equation}%
locally given by\textit{\ }%
\begin{eqnarray*}
E(\lambda ) &=&E_{\sigma }(\lambda )\omega ^{\sigma }\wedge \omega _{0}, \\
E_{\sigma }(\lambda ) &=&\dfrac{\delta \mathcal{L}}{\delta y^{\sigma }}=%
\dfrac{\partial \mathcal{L}}{\partial y^{\sigma }}-d_{i}\dfrac{\partial 
\mathcal{L}}{\partial y_{~i}^{\sigma }}+...+(-1)^{r}d_{i_{1}}...d_{i_{r}}%
\dfrac{\partial \mathcal{L}}{\partial y_{~i_{1}...i_{r}}^{\sigma }}
\end{eqnarray*}%
On-shell, i.e., along critical sections $\gamma ,$ the functions $E_{\sigma
}(\lambda )\circ J^{s+1}\gamma $ vanish. Denoting equality on-shell by $%
\approx $, this is written as:%
\begin{equation*}
E_{\sigma }(\lambda )\circ J^{s+1}\gamma \approx 0.
\end{equation*}%
The Euler-Lagrange form $E(\lambda )$ - though expressed in terms of a
Lepage equivalent $\theta _{\lambda }$ - does not depend on the choice of $%
\theta _{\lambda }.$

\textit{ii)} Let us denote:%
\begin{equation}
\mathcal{J}^{\Xi }:=-h\mathbf{i}_{J^{s}\Xi }\theta _{\lambda }\in \Omega
_{n-1}^{s+1}Y  \label{Noether_current_def}
\end{equation}%
(in local writing, this is: $\mathcal{J}^{\Xi }=J^{i}\omega _{i},~\
J^{i}=J^{i}(x^{k},y^{\sigma },....y_{~i_{1}...i_{s+1}}^{\sigma })$).

Integrating on a compact domain $D\subset X$, the first variation formula
now reads:%
\begin{equation}
\underset{D}{\int }J^{r}\gamma ^{\ast }(\partial _{J^{r}\Xi }\lambda )=%
\underset{D}{\int }J^{s+1}\gamma ^{\ast }\mathbf{i}_{J^{s+1}\Xi }E(\lambda )-%
\underset{\partial D}{\int }J^{s+1}\gamma ^{\ast }\mathcal{J}^{\Xi }.
\label{first_variation_integral}
\end{equation}

The vector field $\Xi \in \mathcal{X}(Y)$ is said to be a \textit{symmetry
generator} for $\lambda $ if $\lambda $ is invariant under the 1-parameter
group of $J^{r}\Xi ;$ this is equivalent to:%
\begin{equation}
\partial _{J^{r}\Xi }\lambda =0.  \label{invariance_condition_1}
\end{equation}%
If $\Xi $ is a symmetry generator for $\lambda ,$ then\ $\mathcal{J}^{\Xi }$
acquires the meaning of \textit{Noether current}.

\textit{Noether's first theorem} states that, if $\Xi \in \mathcal{X}(Y)$ is
a symmetry generator for $\lambda ,$ then, for any Lepage equivalent $\theta
_{\lambda },$ there holds: 
\begin{equation}
J^{s+1}\gamma ^{\ast }d\mathcal{J}^{\Xi }\approx 0.  \label{J_conservation}
\end{equation}

\subsection{Diffeomorphisms and diffeomorphism invariance}

The notion of energy-momentum tensor is tightly related to the invariance of
the matter Lagrangian to the group of spacetime diffeomorphisms $Diff(X).$
But, in order to rigorously define this invariance, some preliminary
discussions are needed.

A diffeomorphism $\Phi :Y\rightarrow Y$ is called, \cite{Krupka-book} an 
\textit{automorphism} of $Y$ if there exists a mapping $\varphi \in Diff(X)$
such that%
\begin{equation}
\pi \circ \Phi =\varphi \circ \pi ;  \label{automorphism}
\end{equation}%
if relation (\ref{automorphism})\ holds, then the automorphism $\Phi $ is
said to \textit{cover }$\varphi .$ An automorphism of $Y\ $is called \textit{%
strict} if it covers the identity of $X.$ We will denote in the following by 
$Aut(Y)$ and $Aut_{s}(Y)$ the sets of automorphisms of $Y$ and,
respectively, strict automorphisms of $Y$.

Passing to infinitesimal generators, any generator $\Xi $ of an automorphism
of $Y$ is a $\pi $-projectable\textit{\ }vector field; in any fibered chart,
it is expressed as:%
\begin{equation}
\Xi =\xi ^{i}(x)\partial _{i}+\Xi ^{\sigma }(x,y)\partial _{\sigma }.
\label{vector_field_Y}
\end{equation}%
Strict automorphisms are generated by \textit{vertical} vector fields $\Xi
=\Xi ^{\sigma }(x,y)\partial _{\sigma }.$

As stated above, we are interested in the way that diffeomorphisms of $X$
affect a Lagrangian $\lambda \in \Omega _{n}^{r}(Y),$ $\lambda =\mathcal{L}%
(x^{i},y^{\sigma },...,y_{~i_{1}...i_{r}}^{\sigma })\omega _{0}.$ But
diffeomorphisms of $X$ do \textit{not} act on the field variables - at
least, not directly; more rigorously stated, $Diff(X$) is not a subgroup of $%
Aut(Y)$ (it is rather a quotient group, \cite{Gotay}, \cite{Forger}, since $%
Diff(X)\simeq Aut(Y)/Aut_{s}(Y)$).

\bigskip

Consequently, in order to be able to speak about diffeomorphism invariance
of the Lagrangian, we need \textit{embeddings\footnote{%
Such liftings canonically exist, for instance, if $Y\overset{\pi }{%
\rightarrow }X$ is a \textit{natural bundle}, \cite{Michor}, i.e., if $Y$ is
obtained from $X$ by applying a covariant functor to the whole category of
differentiable manifolds.} }(or \textit{liftings})\textit{\ }$%
Diff(X)\rightarrow Aut(Y).$

In the following, we will assume that there exists a \textit{canonical}
lifting, given by a group morphism 
\begin{equation}
Diff(X)\rightarrow Aut(Y),~\ \ \varphi \mapsto \Phi ,  \label{lifting_macro}
\end{equation}%
such that, for any $\varphi \in Diff(X):$ $\pi \circ \Phi =\varphi $.

\bigskip

Passing to the infinitesimal level, the canonical lifting gives rise to an $%
\mathbb{R}$-linear module monomorphism%
\begin{equation}
l:\mathcal{X}(X)\rightarrow \mathcal{X}_{P}(Y),~~\xi \mapsto \Xi ,
\label{lift_general}
\end{equation}%
such that $\pi _{\ast |\mathcal{X}_{P}(Y)}\circ l=id_{\mathcal{X}(X)}$. In
local writing, this is given by%
\begin{equation*}
\xi =\xi ^{i}(x^{j})\dfrac{\partial }{\partial x^{i}}~\overset{l}{\mapsto }%
~\Xi =\xi ^{i}(x^{j})\dfrac{\partial }{\partial x^{i}}+\Xi ^{\sigma
}(x^{j},y^{\mu })\dfrac{\partial }{\partial y^{\sigma }}\in \mathcal{X}%
_{P}(Y).
\end{equation*}

\bigskip

The lifting $l$ is required to have the property that $\pi (supp(l(\Xi
)))\subset supp(\xi ),$ where $supp(f)$ denotes the support of a mapping $f$
(defined as the closure of the set $f^{-1}(0)$). As a consequence, \cite%
{Forger}, \cite{Gotay}, in any fibered chart, the local components $\Xi
^{\sigma }$ are expressible as linear combinations of a finite number (say, $%
k$) of partial derivatives of the components $\xi ^{i}:$%
\begin{equation}
\Xi ^{\sigma }=C_{~i}^{\sigma }\xi ^{i}+C_{~i}^{\sigma j}\xi
_{,j}^{i}+....+C_{i}^{\sigma j_{1}...j_{k}}\xi _{,j_{1}...j_{k}}^{i},
\label{Csi_general}
\end{equation}%
where $C_{i}^{\sigma },$ $C_{~i}^{\sigma j},...,C_{i}^{\sigma j_{1}...j_{k}}$
are functions of $(x^{k},y^{\mu })$ only. The number $k\in \mathbb{N}$ is
called the \textit{differential index }(or, simply, the \textit{index})%
\textit{\ }of the lifting.

\bigskip

\textbf{Particular case. }Assume that the index of the lifting $l$ is 1. A
direct calculation shows that, with respect to fibered coordinate changes $%
x^{i}=x^{i}(x^{i^{\prime }}),$ $y^{\sigma }=y^{\sigma }(x^{i^{\prime
}},y^{\sigma ^{\prime }})$, the top degree coefficients $C_{~i}^{\sigma
j}=C_{~i}^{\sigma j}(x^{k},y^{\mu })$\textbf{\ }transform by the rule:%
\begin{equation}
C_{~i^{\prime }}^{\sigma ^{\prime }j^{\prime }}=\dfrac{\partial y^{\sigma
^{\prime }}}{\partial y^{\sigma }}\dfrac{\partial x^{j^{\prime }}}{\partial
x^{j}}\dfrac{\partial x^{i}}{\partial x^{i^{\prime }}}C_{~i}^{\sigma j},
\label{transformation_C}
\end{equation}%
i.e., they define a section\footnote{%
This section represents the so-called \textit{principal symbol\ }of the
differential operator $v\circ l:X(X)\mapsto \mathcal{X}(Y)$ providing, for
any $\xi \in \mathcal{X}(X),$ the $\pi $-vertical component of the lift $%
l(\xi )$.} of the tensor product $TY\otimes TX\otimes T^{\ast }X$.

\bigskip

\textbf{Example. }A typical example of a canonical lifting of (pure) index 1
is obtained if $Y=T^{p,q}(X)$ is the bundle of tensors of type $(p,q)$ over $%
X.$ In this case, any diffeomorphism $\varphi \in Diff(X)$ is canonically
lifted into an automorphism of $T^{p,q}(X)$ by pushforward/pullback.
Denoting the fibered coordinates on $T^{p,q}(X)$ by $%
(x^{i},y_{j_{1}...j_{q}}^{i_{1}...i_{p}}),$ the corresponding mapping $l:%
\mathcal{X}(X)\rightarrow \mathcal{X}(Y)$ is given, \cite{Giachetta1}, by
relation (\ref{vector_field_Y}), with.%
\begin{equation}
\Xi _{j_{1}...j_{q}}^{i_{1}...i_{p}}=\xi
_{,h}^{i_{1}}y_{j_{1}...j_{q}}^{hi_{2}...i_{p}}+...+\xi
_{,h}^{i_{p}}y_{j_{1}...j_{q}}^{i_{1}...i_{p-1}h}-\xi
_{,j_{1}}^{h}y_{hj_{2}...j_{q}}^{i_{1}...i_{p}}-\xi
_{,j_{q}}^{h}y_{j_{1}...j_{q-1}h}^{i_{1}...i_{p}}.  \label{tensor_lifting}
\end{equation}

\bigskip

A differential form $\rho \in \Omega _{k}^{r}Y$ is is said to be \textit{%
diffeomorphism invariant }(or \textit{generally covariant) }if, for any $%
\varphi \in Diff(X),$%
\begin{equation}
J^{r}\Phi ^{\ast }\rho =\rho ,  \label{Diff_invariance}
\end{equation}%
where $\Phi $ denotes the $\varphi $ through the canonical lifting $%
Diff(X)\rightarrow Aut(Y).$

If a differential form $\rho \in \Omega _{k}^{r}Y$ is diffeomorphism
invariant, then the lift $\Xi $ of any vector field $\xi \in \mathcal{X}(X)$
is a symmetry generator for $\lambda .$

\section{The energy-momentum tensor}

Assume, in the following, as in \cite{Giachetta1}, \cite{Giachetta-book},
that the configuration manifold $Y$ is a fibered product%
\begin{equation*}
Y=Y^{(b)}\times _{X}Y^{(m)},
\end{equation*}%
with local coordinates $(x^{i},y^{A},y^{\sigma })$. The coordinates $y^{A}$
will be called \textit{background variables} and $y^{\sigma },$ \textit{%
matter variables. }Consider, on $Y$ a Lagrangian of order $r:$%
\begin{equation}
\lambda =\lambda _{b}+\lambda _{m},  \label{splitting_general}
\end{equation}%
where the \textit{background Lagrangian} $\lambda _{b}$ depends only on $%
y^{A},y_{~i}^{A},...,y_{~i_{1}...i_{r}}^{A}$ and the \textit{matter
Lagrangian}%
\begin{equation*}
\lambda _{m}=\mathcal{L}_{m}(x^{i},y^{A},...,y_{~i_{1}...i_{r}}^{A},y^{%
\sigma },y_{~i}^{\sigma },...,y_{~i_{1}...i_{r}}^{\sigma })\omega _{0}
\end{equation*}%
may depend on all the coordinates on $J^{r}Y.$

\bigskip

We will make the following assumptions:

\textbf{A1. }The lift $l:\mathcal{X}(X)\rightarrow \mathcal{X}(Y),$ $\xi
\mapsto \Xi :=\xi ^{i}(x^{j})\partial _{i}+\Xi ^{A}(x^{j},y^{B})\partial
_{A}+\Xi ^{\sigma }(x^{j},y^{\mu })\partial _{\sigma },$ is of index at at
most 1 in the background variables $y^{A},$ i.e., corresponding to any
fibered chart on $Y$:%
\begin{equation}
\Xi ^{A}=C_{~i}^{A}\xi ^{i}+C_{~i}^{Aj}\xi _{,j}^{i}.  \label{index_1_lift}
\end{equation}

\textbf{A2.} The matter Lagrangian $\lambda _{m}$ is generally covariant.

\textit{Note. }There is no restriction upon the index of the lifting $l$ in
the matter variables $y^{\sigma }$.

\bigskip

The hypothesis A2\textit{\ }implies that, for any vector field $\xi \in 
\mathcal{X}(X),$ $\partial _{J^{r}\Xi }\lambda _{m}=0.$ Let $\theta
_{\lambda _{m}}$ be an arbitrary Lepage equivalent (of order $s$) of $%
\lambda _{m}.$ Using the notations in (\ref{EL_form_coord_free}), (\ref%
{Noether_current_def}), the first variation formula reads:

\begin{equation}
0=J^{s+1}\gamma ^{\ast }\mathbf{i}_{J^{s+1}\Xi }E(\lambda
_{m})-J^{s+1}\gamma ^{\ast }d\mathcal{J}^{\Xi },  \label{first_var_J}
\end{equation}%
for \textit{any} section $\gamma :=(\gamma ^{(b)},\gamma
^{(m)}):X\rightarrow Y,$ locally represented as $\gamma
(x^{i})=(x^{i},y^{A}(x^{i}),y^{\sigma }(x^{i})).$

The Euler-Lagrange form $E(\lambda _{m})=p_{1}d\theta _{\lambda _{m}}$
splits into background\ and matter components as:%
\begin{equation}
E(\lambda _{m})=E^{(b)}(\lambda _{m})+E^{(m)}(\lambda _{m}),
\label{EL_form_splitting}
\end{equation}%
where:%
\begin{equation*}
E^{(b)}(\lambda _{m})=p_{1}^{(b)}d\theta _{\lambda _{m}}=\dfrac{\delta 
\mathcal{L}_{m}}{\delta y^{A}}\omega ^{A}\wedge \omega _{0},~E^{(m)}(\lambda
_{m})=\ p_{1}^{(m)}d\theta _{\lambda _{m}}=\dfrac{\delta \mathcal{L}_{m}}{%
\delta y^{\sigma }}\omega ^{\sigma }\wedge \omega _{0}.
\end{equation*}

The background component $E^{(b)}(\lambda _{m})$ in (\ref{EL_form_splitting}%
) will be used as a "raw material"\ for the energy-momentum tensor.

\begin{definition}
The \textbf{energy-momentum}\textit{\ }\textbf{source form}\footnote{%
In \cite{Krupka-em-tensors}, the source form $\tau $ is called the
energy-momentum \textit{tensor. }Still, we will prefer here a slightly
different terminology, for reasons which will be clarified below.} of $%
\lambda _{m}$ is the Euler-Lagrange form of $\lambda _{m}$ with respect to
the background variables: $\tau :=E^{(b)}(\lambda _{m}).$
\end{definition}

In local writing, 
\begin{equation*}
\tau =\tau _{A}\omega ^{A}\wedge \omega _{0}\in \Omega _{n+1}^{s+1}(Y),
\end{equation*}%
where:%
\begin{equation}
\tau _{A}=\dfrac{\delta \mathcal{L}_{m}}{\delta y^{A}}.  \label{def_tau}
\end{equation}

\textbf{Remark. }In \cite{Krupka-book}, p. 118, it is proven the following
result. For any automorphism $\Phi \in Aut(Y)$ and any Lagrangian $\lambda
\in \Omega _{n}^{r}Y,$%
\begin{equation}
J^{s+1}\Phi ^{\ast }E(\lambda )=E(J^{r}\Phi ^{\ast }\lambda );
\label{invar_EL_form}
\end{equation}%
in particular, if $\lambda $ is $\Phi $-invariant, then so is its
Euler-Lagrange form $E(\lambda ).$ Using this result, we find out that, if
the Lagrangian $\lambda _{m}\in \Omega _{n}^{r}Y$ is $\Phi $-invariant, then 
$\tau $ is also $\Phi $-invariant.

\bigskip

Taking into account the splitting (\ref{EL_form_splitting}), the first
variation formula (\ref{first_var_J}) becomes:%
\begin{equation}
0=J^{s+1}\gamma ^{\ast }\left[ \mathbf{i}_{J^{s+1}\Xi }E^{(m)}(\lambda _{m})+%
\mathbf{i}_{J^{s+1}\Xi }\tau \right] -J^{s+1}\gamma ^{\ast }d\mathcal{J}%
^{\Xi }.  \label{first_variation_split}
\end{equation}%
On-shell for the matter variables $y^{\sigma },$ i.e., when the matter
component $\gamma ^{(m)}$ of the section $\gamma =(\gamma ^{(b)},\gamma
^{(m)}):X\rightarrow Y,$ $(x^{i})\mapsto (y^{A}(x^{i}),y^{\sigma }(x^{i}))$
is critical for $\lambda _{m},$ we obtain:%
\begin{equation}
J^{s+1}\gamma ^{\ast }(\mathbf{i}_{J^{s+1}\Xi }\tau )-J^{s+1}\gamma ^{\ast }d%
\mathcal{J}^{\Xi }\approx _{y^{\sigma }}0,  \label{first_variation_1}
\end{equation}%
or, integrating on a compact domain $D\subset X$ (and applying the
horizontalization operator):%
\begin{equation}
\underset{D}{\int }J^{s+2}\gamma ^{\ast }(h\mathbf{i}_{J^{s+1}\Xi }\tau )-%
\underset{\partial D}{\int }J^{s+1}\gamma ^{\ast }d\mathcal{J}^{\Xi }\approx
_{y^{\sigma }}0.  \label{first_variation_2}
\end{equation}

\bigskip

Let us investigate more closely the volume term in the above intergral.

\begin{lemma}
If the embedding $l:\mathcal{X}(X)\rightarrow \mathcal{X}(Y),$ $\xi \mapsto
\Xi ,$ is of index at most 1 in the background variables $y^{A},$ then there
uniquely exist the $\mathcal{F}(X)$-linear mappings $\mathcal{B}:\mathcal{X}%
(X)\rightarrow \Omega _{n}^{s+2}Y$ and $\mathcal{T}:\mathcal{X}%
(X)\rightarrow \Omega _{n-1}^{s+1}Y$ with $\pi ^{r}$-horizontal values,
satisfying: 
\begin{equation}
h\mathbf{i}_{J^{s+1}\Xi }\tau =\mathcal{B}(\xi )+hd(\mathcal{T}(\xi )),~\ \
\ \forall \xi \in \mathcal{X}(X).  \label{def_BT}
\end{equation}
\end{lemma}

\begin{proof}
We first define $\mathcal{B}$ and $\mathcal{T}$ locally. Take an arbitrary
fibered chart $(V^{s+2},\psi ^{s+2})$ on $J^{s+2}Y$, with $U=\pi
^{s+2}(V^{s+2}).$ In this chart, $h\mathbf{i}_{J^{s+1}\Xi }\tau $ is
expressed as:%
\begin{equation}
h\mathbf{i}_{J^{s+1}\Xi }\tau =(\tilde{\Xi}^{A}\tau _{A})\omega _{0},~\ \ \
\ \ \tilde{\Xi}^{A}=\Xi ^{A}-y_{~i}^{A}\xi ^{i}.  \label{EL_term}
\end{equation}%
With $\Xi ^{A}$ from (\ref{index_1_lift}), we find 
\begin{eqnarray}
&&h\mathbf{i}_{J^{s+1}\Xi }\tau =\{(C_{~i}^{A}-y_{~i}^{A})\tau _{A}\xi
^{i}+C_{~i}^{Aj}\tau _{A}\xi _{~,j}^{i}\}\omega _{0}  \label{h_tau} \\
&=&\{[(C_{~i}^{A}-y_{~i}^{A})\tau _{A}-d_{j}(C_{~i}^{Aj}\tau _{A})]\xi
^{i}+d_{j}(C_{~i}^{Aj}\tau _{A}\xi ^{i})\}\omega _{0}.  \notag
\end{eqnarray}%
Denoting 
\begin{eqnarray}
\mathcal{B}(\xi ) &=&[(C_{~i}^{A}-y_{~i}^{A})\tau _{A}-d_{j}(C_{~i}^{Aj}\tau
_{A})]\xi ^{i}\omega _{0},  \label{B_local} \\
\mathcal{T}(\xi ) &:&=(C_{~i}^{Aj}\tau _{A}\xi ^{i})\omega _{j},~\ \ ~\ \ \
\forall \xi \in \mathcal{X}(U),  \label{T_local1}
\end{eqnarray}%
we obtain two linear mappings $\mathcal{B}:\mathcal{X}(U)\mapsto \Omega
_{n}^{s+2}(Y)$ and $\mathcal{T}:\mathcal{X}(U)\mapsto \Omega
_{n-1}^{s+1}(Y), $ having horizontal values and obeying (\ref{def_BT}).

The uniqueness of the (momentarily, locally defined) mappings $\mathcal{B}$
and $\mathcal{T}$ can be established as follows. Assume that the mappings $%
\mathcal{\tilde{B}}:\mathcal{X}(U)\mapsto \Omega _{n}^{s+2}(Y),$ $\mathcal{%
\tilde{T}}:\mathcal{X}(U)\mapsto \Omega _{n-1}^{s+1}(Y)$ also obey the above
properties. Then, for any $\xi \in \mathcal{X}(X),$ we have:%
\begin{equation*}
0=(\mathcal{B-\tilde{B}})(\xi )+hd[(\mathcal{T-\tilde{T}})(\xi )].
\end{equation*}%
As all, these mappings have horizontal values, they can be expressed as: $%
\mathcal{B}(\xi )=\mathcal{B}_{i}\xi ^{i}\omega _{0}$, $\mathcal{\tilde{B}}%
(\xi )=\mathcal{\tilde{B}}_{i}\xi ^{i}\omega _{0},$ $\mathcal{T}(\xi )=%
\mathcal{T}_{~i}^{j}\xi ^{i}\omega _{j},$ $\mathcal{\tilde{T}}(\xi )=%
\mathcal{\tilde{T}}_{~i}^{j}\xi ^{i}\omega _{j}.$ Substituting into the
above equality,%
\begin{equation*}
0=[(\mathcal{B}_{i}\mathcal{-\tilde{B}}_{i})+d_{j}(\mathcal{T}_{~i}^{j}%
\mathcal{-\tilde{T}}_{~i}^{j})]\xi ^{i}+(\mathcal{T}_{~i}^{j}\mathcal{-%
\tilde{T}}_{~i}^{j})\xi _{,j}^{i}.
\end{equation*}%
As this relation holds for \textit{any }$\xi ,$ we obtain $\mathcal{T}%
_{~i}^{j}\mathcal{-\tilde{T}}_{~i}^{j}=0$ and $\mathcal{B}_{i}\mathcal{-%
\tilde{B}}_{i}=0.$ Therefore, $\mathcal{B=\tilde{B}}$ and $\mathcal{T-\tilde{%
T}}.$

Now, take two fibered chart domains $V^{s+2},V^{s+2^{\prime }}\subset
J^{s+2}Y,$ with $U=\pi ^{s+2}(V^{s+2}),U^{\prime }=\pi ^{s+2}(V^{s+2^{\prime
}})$ and an arbitrary vector field $\xi \in \mathcal{X}(U\cap U^{\prime }).$
We denote by $\mathcal{T}$ and $\mathcal{T}^{\prime }$ the mappings
corresponding by (\ref{T_local1}) to the two domains. Taking into account
the rules of transformation (\ref{transf_source_form}), (\ref%
{transformation_C}), (\ref{transf_omega0}) of $\tau _{A},$ $C_{~i}^{Aj}$ and 
$\omega _{i}$, a brief calculation shows that:%
\begin{equation*}
\mathcal{T}(\xi )=(C_{~i}^{Aj}\tau _{A}\xi ^{i})\omega _{j}=(C_{~i^{\prime
}}^{A^{\prime }j^{\prime }}\tau _{A^{\prime }}\xi ^{i^{\prime }})\omega
_{j^{\prime }}=\mathcal{T}^{\prime }(\xi )\ 
\end{equation*}%
i.e., the mapping $\mathcal{T}$ can be defined globally on $\mathcal{X}(X)$.
As a consequence of (\ref{def_BT}), $\mathcal{B}$ is also globally well
defined.
\end{proof}

Therefore, it makes sense

\begin{definition}
\label{def_T}The \textbf{energy-momentum tensor }of $\lambda _{m}$ is the
mapping:%
\begin{equation}
\mathcal{T}:\mathcal{X}(X)\rightarrow \Omega _{n-1}^{s+1}(Y),~\ \ \xi
\mapsto \mathcal{T}(\xi ),  \label{T_def}
\end{equation}%
uniquely defined by the splitting 
\begin{equation}
h\mathbf{i}_{J^{s+1}\Xi }\tau =\mathcal{B}(\xi )+hd(\mathcal{T}(\xi )),
\label{tau_splitting}
\end{equation}%
where the mappings $\mathcal{T}:\mathcal{X}(X)\rightarrow \Omega
_{n-1}^{s+1}(Y),$ $\xi \mapsto \mathcal{T}(\xi )$ and $\mathcal{B}:\mathcal{X%
}(X)\rightarrow \Omega _{n}^{s+2}(Y),$ $\xi \mapsto \mathcal{B}(\xi )$ are $%
\mathcal{F}(X)$-linear and have $\pi ^{r}$-horizontal values.
\end{definition}

In the following, we will call the mapping $\mathcal{B}$ defined by (\ref%
{tau_splitting}), the \textit{balance function.}

\begin{corollary}
In local writing, the energy-momentum tensor is given by:%
\begin{equation}
\xi ^{i}\partial _{i}\mapsto \mathcal{T}(\xi )=\mathcal{T}_{~i}^{j}\xi
^{j}\omega _{j},  \label{local_T}
\end{equation}%
where%
\begin{equation}
\mathcal{T}_{~i}^{j}=C_{~i}^{Aj}\tau _{A}=C_{~i}^{Aj}\dfrac{\delta \mathcal{L%
}_{m}}{\delta y^{A}}.  \label{T_comps}
\end{equation}%
With respect to fibered coordinate changes on $J^{s+1}Y$, the functions $%
\mathcal{T}_{~i}^{j}$ obey the rule:%
\begin{equation}
\mathcal{T}_{~i}^{j}=\dfrac{\partial x^{j}}{\partial x^{j^{\prime }}}\dfrac{%
\partial x^{i^{\prime }}}{\partial x^{i}}\det (\dfrac{\partial x^{k^{\prime
}}}{\partial x^{k}})\mathcal{T}_{~i^{\prime }}^{j^{\prime }}.
\label{trans_T_comps}
\end{equation}
\end{corollary}

\textbf{Remarks. }

1)\ The local expression (\ref{T_comps}) of $\mathcal{T}$ coincides, up to a
minus sign (and a pullback by sections of $J^{s+1}Y$), to the one found by
Gotay and Marsden, \cite{Gotay}, as a result of a different, Noether-type
construction. The reason for our choice of the sign in (\ref{def_BT}) (i.e.,
also in (\ref{T_comps})) is that it gives, in the case of purely metric
backgrounds, the usual Hilbert energy-momentum tensor, with a correct sign.

2)\ Here, the energy-momentum tensor $\mathcal{T}$ is regarded as a
geometric object on the jet bundle $J^{s+1}Y$ and not on the base manifold $%
X $. For each section $\gamma \in \Gamma (Y),$ a corresponding linear
mapping $\mathcal{T}_{(\gamma )}:\mathcal{X}(X)\mapsto \Omega _{n-1}(X)$ can
be obtained from $\mathcal{T}$ by pullback: $\mathcal{T}_{(\gamma )}(\xi
):=J^{s+1}\gamma ^{\ast }\mathcal{T}(\xi ),$ $\forall \xi \in \mathcal{X}%
(X). $ This way, the energy-momentum tensor "on $X$" $\mathcal{T}_{(\gamma )}
$ is regarded as an element of $\Omega _{1}(X)\otimes \Omega _{n-1}(X).$

3)\ The rule of transformation (\ref{trans_T_comps}) of its coefficients is
due to the fact that $\mathcal{T}$ is expressed in the non-invariant local
basis $dx^{i}\otimes \omega _{j}.$ If $\omega _{j}$ are replaced with an
invariant basis for horizontal $n-1$ forms, then the corresponding local
components of $\mathcal{T}$ will obey a tensor-type transformation rule.

4)\ Equivalent Lagrangians give rise to the same energy-momentum tensor $%
\mathcal{T};$ this is a consequence of the fact that $\mathcal{T}$ is a
combination of Euler-Lagrange expressions of $\lambda _{m}.$

5)\ For energy-momentum source forms $\tau $ whose order in the background
variables $y^{A}$ does not exceed 1, a complete characterization (as a
polynomial in these variables) is available, \cite{Krupka-em-tensors}, \cite%
{inv-problem-book}. This gives us immediately a general characterization of
the corresponding energy-momentum tensors $\mathcal{T}.$

\section{Properties of the energy-momentum tensor}

Using Definition \ref{def_T}, we will prove:

\begin{theorem}
\label{Hilbert_SEM_prop}(i) (The energy-momentum balance law): In any
fibered chart of $J^{s+2}Y,$ there hold the relations:%
\begin{equation}
\left( d_{j}\mathcal{T}_{~i}^{j}-(C_{~i}^{A}-y_{~i}^{A})\tau _{A}\right)
\circ J^{s+2}\gamma \approx _{y^{\sigma }}0,~\ \ i=1,...,n.
\label{SEM_balance_relation}
\end{equation}%
where $\approx _{y^{\sigma }}$ means equality on-shell for the matter
component $\gamma _{m}:(x^{i})\mapsto (y^{\sigma }(x^{i}))$ of the section $%
\gamma =(\gamma ^{(b)},\gamma ^{(m)}):X\rightarrow Y.$

(ii) For any vector field $\xi \in \mathcal{X}(X)$ and any compact domain $%
D\subset X:$ 
\begin{equation}
\underset{\partial D}{\int }J^{s+1}\gamma ^{\ast }\mathcal{T}(\xi )\approx
_{y^{\sigma }}\underset{\partial D}{\int }J^{s+1}\gamma ^{\ast }\mathcal{J}%
^{l(\xi )},  \label{correct_T_variational}
\end{equation}%
where $l:\mathcal{X}(X)\rightarrow \mathcal{X}(Y)$ denotes the canonical
lift.
\end{theorem}

\begin{proof}
Take an arbitrary vector field $\xi \in \mathcal{X}(X)$ and denote $\Xi
:=l(\xi ).$ Using the splitting (\ref{tau_splitting}) together with the
integral first variation formula (\ref{first_variation_2}), we get:%
\begin{equation}
0\approx _{y^{\sigma }}\underset{D}{\int }J^{s+2}\gamma ^{\ast }\mathcal{B}%
(\xi )+\underset{\partial D}{\int }J^{s+1}\gamma ^{\ast }(\mathcal{T}(\xi )-%
\mathcal{J}^{\Xi }).  \label{global_T_variation}
\end{equation}%
Equation (\ref{global_T_variation}) holds, in particular, for any vector
field $\xi $ with support contained in $D,$ i.e., $\xi _{|\partial D}=0.$ It
follows that%
\begin{equation}
\underset{D}{\int }J^{s+2}\gamma ^{\ast }\mathcal{B}(\xi )\approx
_{y^{\sigma }}0;  \label{balance_law_global}
\end{equation}%
using the local writing (\ref{B_local}) of the balance function $\mathcal{B}$%
, this gives \textit{(i).} Then, choosing $\xi $ with $\xi _{|\partial
D}\not=0,$ and using \textit{(i)}, together with (\ref{global_T_variation}),
we find \textit{(ii)}.
\end{proof}

\textbf{Remarks. }

1) Property \textit{(ii)} tells us that $\mathcal{T}(\xi )$ coincides
on-shell with the "improved Noether current" given by the diffeomorphism
invariance of $\lambda _{m}$.

2)\ The order of the Lagrangian $\lambda _{m}$ either in the background
variables or in the matter ones was completely irrelevant in deducing the
above results. That is, they are valid also in the case of non-minimal
coupling between the background and the matter variables.

3)\ Since it is defined as a linear combination of Euler-Lagrange
expressions $\tau _{A},$ the energy-momentum tensor $\mathcal{T}$ does not
depend on the choice of the Lepage equivalent $\theta _{\lambda _{m}}$ (on
the contrary, Noether currents do depend on the choice of the Lepage
equivalent).

4)\ If we apply the above construction to the background Lagrangian $\lambda
_{b},$ which only depends on the background variables, then the
corresponding energy-momentum balance relations (\ref{SEM_balance_relation})
will be identically satisfied for \textit{any} section $\gamma \in \Gamma
(Y);$ as we will in the next section, in the particular case of general
relativity, these are actually, the contracted Bianchi identities.

\begin{theorem}
(Gauge invariance of $\mathcal{T}$): If a strict automorphism $\Phi \in
Aut_{s}(Y^{(b)}\times _{X}Y^{(m)})$ is a symmetry of $\lambda $ acting
trivially on the background manifold $Y^{(b)},$ then:%
\begin{equation}
J^{s+1}\Phi ^{\ast }\mathcal{T}(\xi )=\mathcal{T}(\xi ),\forall \xi \in 
\mathcal{X}(X).  \label{gauge_invariance_T}
\end{equation}
\end{theorem}

\begin{proof}
According to the hypothesis, $\Phi =(id_{X},id_{Y^{(b)}},\Phi ^{(m)});$
locally, $\Phi :(x^{i},y^{\sigma },y^{A})\mapsto (x^{i},f^{\sigma
}(x^{i},y^{\mu }),y^{A}).$

Since $\Phi $ is a symmetry of $\lambda ,$ we deduce that $J^{s+1}\Phi
^{\ast }\tau =\tau .$ Moreover, a brief calculation using the local
expression of $\Phi $ shows that: $J^{s+2}\Phi ^{\ast }(h\mathbf{i}%
_{J^{s+1}\Xi }\tau )=h\mathbf{i}_{J^{s+1}\Xi }\tau .$ Further, using the
uniqueness of the decomposition (\ref{tau_splitting}) of $h\mathbf{i}%
_{J^{s+1}\Xi }\tau $, we find that%
\begin{equation*}
J^{s+2}\Phi ^{\ast }\mathcal{B}(\xi )=\mathcal{B}(\xi ),~\ J^{s+1}\Phi
^{\ast }\mathcal{T}(\xi )=\mathcal{T}(\xi );
\end{equation*}%
the latter equality proves the statement.
\end{proof}

\section{The case of metric and tensor backgrounds}

We will study, in the following, the case when the background variables
consist of a metric and (optionally), some other tensor quantity. In this
case, it is useful to write the energy-momentum balance law (\ref%
{SEM_balance_relation}) in a manifestly covariant form, using Levi-Civita
covariant derivatives.

Denoting by $Met(X)$ the bundle of metrics, defined as the set of all
symmetric nondegenerate tensors of type (2,0)\ on $X$, the background
manifold becomes:%
\begin{equation*}
Y^{(b)}=Met(X)\times _{X}T^{p,q}(X)
\end{equation*}%
and, accordingly, $Y=Met(X)\times _{X}T^{p,q}(X)\times _{X}Y^{(m)}.$ With
the notations in the previous sections, the background variables are $%
y^{A}\in \{g^{jk},y_{j_{1}...j_{q}}^{i_{1}...i_{p}}\}$ (alternatively, one
can regard $Met(X)$ as a set of tensors of type (0,2) and use $g_{jk}$ as
coordinates). In the following, $dV_{g}=\sqrt{\left\vert \det g\right\vert }%
\omega _{0}$ will mean the Riemannian volume form on $X.$

Assume that the Lagrangian $\lambda $ is \textit{natural}, i.e., 
\begin{equation*}
\lambda _{m}=L_{m}dV_{g},
\end{equation*}%
where $L_{m}=L_{m}(x^{i},y^{\sigma },g^{jk},y_{~i}^{\sigma
},g_{~~i}^{jk},...y_{~i_{1}...i_{r}}^{\sigma },g_{~~i_{1}...i_{r}}^{jk})$ is
a differential invariant (also commonly called in the literature, a \textit{%
scalar}); with the notations in the above sections, we have: $\lambda _{m}=%
\mathcal{L}_{m}\omega _{0},$ where: 
\begin{equation}
\mathcal{L}_{m}=L_{m}\sqrt{\left\vert \det g\right\vert }.  \label{L_m_GR}
\end{equation}%
As any such Lagrangian is generally covariant, we can apply the above scheme.

\bigskip

The energy-momentum source form $\tau $ can be expressed as:%
\begin{equation}
\tau =\tau _{A}\omega ^{A}\wedge \omega _{0}=:\mathfrak{T}_{A}\omega
^{A}\wedge dV_{g},  \label{invar_expression_tau}
\end{equation}%
where%
\begin{equation}
\mathfrak{T}_{A}:=\dfrac{\tau _{A}}{\sqrt{\left\vert \det g\right\vert }}=%
\dfrac{1}{\sqrt{\left\vert \det g\right\vert }}\dfrac{\delta \mathcal{L}_{m}%
}{\delta y^{A}};  \label{T_A}
\end{equation}%
Accordingly, the energy-momentum tensor $\mathcal{T}:\mathcal{X}%
(X)\rightarrow \Omega _{n-1}^{s+1}(Y),$ $\xi \mapsto \mathcal{T}(\xi )=%
\mathcal{T}_{~i}^{j}\xi ^{i}\omega _{j}$ can be written, in any fibered
chart, as:%
\begin{equation*}
\mathcal{T}(\xi )=:T_{~i}^{j}\xi ^{i}\sqrt{\left\vert \det g\right\vert }%
\omega _{j},
\end{equation*}%
with:%
\begin{equation}
T_{~i}^{j}=C_{~i}^{Aj}\mathfrak{T}_{A}=\dfrac{1}{\sqrt{\left\vert \det
g\right\vert }}\mathcal{T}_{~i}^{j}.  \label{e-m_tensor}
\end{equation}%
Using (\ref{trans_T_comps}), we can see that, with respect to fibered
coordinate changes, $T_{~i}^{j}$ obey a tensor-type transformation rule: $%
T_{~i^{\prime }}^{j^{\prime }}=\dfrac{\partial x^{j^{\prime }}}{\partial
x^{j}}\dfrac{\partial x^{i}}{\partial x^{i^{\prime }}}T_{~i}^{j}.$

\bigskip

The canonical lifting%
\begin{equation*}
l^{(b)}:\mathcal{X}(X)\rightarrow \mathcal{X}(Y^{(b)}),~\ \ \xi ^{i}\dfrac{%
\partial }{\partial x^{i}}\mapsto \xi ^{i}\dfrac{\partial }{\partial x^{i}}%
+(C_{~i}^{A}\xi ^{i}+C_{~i}^{Aj}\xi _{,j}^{i})\dfrac{\partial }{\partial
y^{A}}
\end{equation*}%
is given by (\ref{tensor_lifting}), i.e., $C_{~i}^{A}=0$ and $C_{~i}^{Aj}\in
\{C_{~~i}^{(jk)h},C_{(j_{1}....j_{q})i}^{(i_{1}...i_{p})j}\}$ are as follows:%
\begin{eqnarray}
&&C_{~~i}^{(jk)h}=\delta _{i}^{j}g^{kh}+\delta _{i}^{k}g^{jh},
\label{C_metric} \\
&&C_{(j_{1}....j_{q})i}^{(i_{1}...i_{p})j}=\delta
_{i}^{i_{1}}y_{j_{1}...j_{q}}^{ji_{2}...i_{p}}+...+\delta
_{i}^{i_{p}}y_{j_{1}...j_{q}}^{i_{1}...i_{p-1}j}-\delta
_{j_{1}}^{j}y_{ij_{2}...j_{q}}^{i_{1}...i_{p}}-\delta
_{j_{q}}^{j}~y_{j_{1}...j_{q-1}i}^{i_{1}...i_{p}}.  \label{C_tensor1}
\end{eqnarray}%
The energy-momentum tensor components (\ref{e-m_tensor}) are: 
\begin{equation}
T_{~i}^{j}=2g^{jh}\mathfrak{T}_{hi}+(y_{j_{1}...j_{q}}^{ji_{2}...i_{p}}%
\mathfrak{T}%
_{ii_{2}...i_{p}}^{j_{1}...j_{q}}+....-y_{j_{1}...j_{q-1}i}^{i_{1}...i_{p}}%
\mathfrak{T}_{i_{1}i_{2}...i_{p}}^{j_{1}...j_{q-1}j}).
\label{tensor_background_T}
\end{equation}%
A direct calculation using the relation $d_{j}\sqrt{\left\vert \det
g\right\vert }=\Gamma _{~ji}^{i}\sqrt{\left\vert \det g\right\vert }$ (where 
$\Gamma _{~jk}^{i}$ denote the formal\footnote{%
Here, \textit{formal }means the fact that the Christoffel symbols $\Gamma
_{~jk}^{i}$ (defined by the usual formula)\ are regarded as functions on $%
J^{1}Met(X);$ only when evaluated on a section $(x^{i})\mapsto
(g^{jk}(x^{i})),$ they become the usual Christoffel symbols on $X.$}
Christoffel symbols of $g$) shows that the energy-momentum balance law (\ref%
{SEM_balance_relation}) can be written in terms of formal Levi-Civita
covariant derivatives $_{;i}$, as:%
\begin{equation*}
\xi ^{i}[y_{~;i}^{A}\mathfrak{T}_{A}+T_{~i;j}^{j}]\approx _{y^{\sigma }}0.
\end{equation*}%
Taking into account that $g_{~~;i}^{jk}=0,$ in the above relation, the $%
Met(X)$ part of the expression $y_{~;i}^{A}\mathfrak{T}_{A}$ vanishes, i.e., 
$y_{~;i}^{A}\mathfrak{T}_{A}=y_{j_{1}...j_{q};i}^{i_{1}...i_{p}}\mathfrak{T}%
_{i_{1}...i_{p}}^{j_{1}...j_{q}}.$ In other words:

\begin{theorem}
(manifestly covariant version of energy-momentum balance law): If the
background manifold is $Y^{(b)}=Met(X)\times _{X}T^{p,q}(X),$ then, for any
natural matter Lagrangian $\lambda _{m}=\mathcal{L}_{m}\omega _{0}\in \Omega
_{n}^{r}(Y^{(b)}\times _{X}Y^{(m)})$ and for any section $\gamma
:X\rightarrow Y:$%
\begin{equation}
(y_{~;i}^{A}\mathfrak{T}_{A}+T_{~i;j}^{j})\circ J^{s+2}\gamma \approx
_{y^{\sigma }}0,~\ \ i=1,...,n,  \label{covariant_cons_law_general}
\end{equation}%
where: $\mathfrak{T}_{A}=\dfrac{1}{\sqrt{\left\vert \det g\right\vert }}%
\dfrac{\delta \mathcal{L}_{m}}{\delta y^{A}},\ A=\left(
_{j_{1}...j_{q}}^{i_{1}...i_{p}}\right) ,$ semicolons denote Levi-Civita
covariant derivatives and $\approx _{y^{\sigma }}$means equality on-shell
for the matter variables $y^{\sigma }$.
\end{theorem}

Let us investigate, in the following, two particular cases.

\subsection{Purely metric theories\textbf{\ }}

Assume that the only background variable is a metric, i.e., $Y^{(b)}=Met(X),$
$y^{A}=g^{jk}.$ In this case, the energy-momentum source form is 
\begin{equation*}
\tau =\mathfrak{T}_{hl}\omega ^{hl}\wedge dV_{g},~\ \ \ \ \ \mathfrak{T}%
_{hl}=\dfrac{1}{\sqrt{\left\vert \det g\right\vert }}\dfrac{\delta \mathcal{L%
}_{m}}{\delta g^{hl}}.\ 
\end{equation*}%
The energy-momentum tensor $T:\mathcal{X}(X)\rightarrow \Omega
_{n-1}^{s+1}(Y)$ is locally given by:%
\begin{equation}
T_{~i}^{j}=C_{~~i}^{(hl)j}\mathfrak{T}_{hl}=2g^{hj}\mathfrak{T}_{hi}.
\label{Hilbert_SEM_density}
\end{equation}

Lowering indices by $g,$ we get:

\begin{proposition}
If the only background variable is a metric tensor $g^{ij}$, then the
energy-momentum tensor $\mathcal{T}$ is given by%
\begin{equation*}
T_{ij}=\dfrac{2}{\sqrt{\left\vert \det g\right\vert }}\dfrac{\delta \mathcal{%
L}_{m}}{\delta g^{ij}}.
\end{equation*}
\end{proposition}

The energy-momentum balance law (\ref{covariant_cons_law_general}) reads: 
\begin{equation}
T_{~i;j}^{j}\circ J^{s+2}\gamma \approx _{y^{\sigma }}0.
\label{covariant_cons_law_1}
\end{equation}

\bigskip

Applying the same algorithm to the Hilbert Lagrangian $\lambda _{g}=RdV_{g},$
the corresponding "energy-momentum tensor" (\ref{Hilbert_SEM_density})\ is
the Einstein tensor $\mathcal{E}=E_{~i}^{j}\sqrt{\left\vert \det
g\right\vert }dx^{i}\otimes \omega _{j}$; the covariant conservation law (%
\ref{covariant_cons_law_1}) gives the contracted Bianchi identities:%
\begin{equation*}
E_{~i;j}^{j}\circ J^{s+2}\gamma =0.
\end{equation*}

\subsection{Metric-affine theories}

In a metric-affine theory, the background variables are, \textit{a priori},
a metric and a connection. Hence, at first sight, the background manifold is 
$Y^{(b)}=Met(X)\times _{X}Conn(X)$ - and the canonical lift $l$ is, in this
case, of index 2, \cite{Giachetta1}, i.e., we cannot apply it the above
considerations. But this problem can be overcome if we consider, instead of
connections, \textit{distortion tensors}%
\begin{equation}
N_{~jk}^{i}=K_{~jk}^{i}-\Gamma _{~jk}^{i},  \label{distortion_tensor}
\end{equation}%
giving the difference between the connection (with coefficients $K_{~jk}^{i}$%
) of the theory and the Levi-Civita connection of the metric. That is, we
can consider as our background manifold%
\begin{equation*}
Y^{(b)}=Met(X)\times _{X}T^{1,2}(X),
\end{equation*}%
with fibered coordinates $(x^{i},g^{ij},N_{~jk}^{i})$. Since both factors
are tensor spaces, we can use the canonical lift (\ref{tensor_lifting}), of
index 1 in the background variables. The corresponding coefficients $%
C_{~i}^{Aj}\in \{C_{~~i}^{(hl)j},C_{(hl)i}^{mj}\}$ are: 
\begin{equation*}
C_{~~i}^{(hl)j}=\delta _{i}^{l}g^{hj}+\delta _{i}^{h}g^{jl},\ \ \ \
C_{(hl)i}^{mj}=\delta _{i}^{m}N_{~hl}^{j}-\delta _{h}^{j}N_{~il}^{m}-\delta
_{l}^{j}N_{~hi}^{m}.
\end{equation*}%
The energy-momentum tensor components $T_{~i}^{j}:=C_{~i}^{Aj}\mathfrak{T}%
_{A}$ are obtained as:%
\begin{equation}
T_{~i}^{j}=2\mathfrak{T}_{~i}^{j}+(\mathfrak{T}_{~~i}^{hl}N_{~hl}^{j}-%
\mathfrak{T}_{~~m}^{jl}N_{~il}^{m}-\mathfrak{T}_{~~m}^{hj}N_{~hi}^{m}),
\label{MA_energy-momentum}
\end{equation}%
where:%
\begin{equation*}
\mathfrak{T}_{ij}=\dfrac{1}{\sqrt{\left\vert \det g\right\vert }}\dfrac{%
\delta \mathcal{L}_{m}}{\delta g^{ij}},\text{~\ \ \ }\mathfrak{T}_{~~i}^{jk}=%
\dfrac{1}{\sqrt{\left\vert \det g\right\vert }}\dfrac{\delta \mathcal{L}_{m}%
}{\delta N_{~jk}^{i}}=\dfrac{\delta L_{m}}{\delta N_{~kh}^{j}}.
\end{equation*}

\textbf{Remark. }The energy-momentum tensor $T$ contains a contribution from
the metric (which is the symmetric, Hilbert energy-momentum tensor $\overset{%
Met}{T}_{ij}=2\mathfrak{T}_{ij}$) plus a term given by the distortion tensor 
$N.$ That is, $T_{ij}$ is, generally, non-symmetric.

Applying (\ref{covariant_cons_law_general}), we find:

\begin{proposition}
In a metric-affine theory, the energy-momentum tensor (\ref%
{MA_energy-momentum}) obeys the balance law:%
\begin{equation}
(T_{~i;j}^{j}+N_{~kh;i}^{j}\dfrac{\delta L_{m}}{\delta N_{~kh}^{j}})\circ
J^{s+2}\gamma \approx _{y^{\sigma }}0,  \label{MA_conservation_law}
\end{equation}%
where semicolons denote Levi-Civita covariant derivatives.
\end{proposition}

The precise form of the matter Lagrangian function $L_{m}$, as well as its
order in either the background or the dynamical variables is irrelevant.

Using in the above procedure, instead of $\lambda _{m},$ the background
Lagrangian $\lambda _{b}=L_{b}\sqrt{\left\vert \det g\right\vert }\omega
_{0},$ the obtained energy-momentum balance law (\ref{MA_conservation_law})
becomes an identity.

\section{Energy-momentum tensor and variational completion}

Given a source form $\varepsilon =\varepsilon _{\sigma }\omega ^{\sigma
}\wedge \omega _{0}$ on a fibered manifold $Y,$ a (local) \textit{%
variational completion, \cite{canonical-var-compl}, }of $\varepsilon $ is a
source form $\kappa $ on an open set $W\subset Y,$ with the property that $%
\varepsilon +\kappa $ is variational. In particular, the \textit{canonical
variational completion} of $\varepsilon $ is the source form $\kappa
(\varepsilon )$ given by the difference between the Euler-Lagrange form of
the \textit{Vainberg--Tonti Lagrangian}%
\begin{equation}
\mathit{\ }\lambda _{\varepsilon }=\mathcal{L}_{\varepsilon }\omega _{0},~\
\ \ \mathcal{L}_{\varepsilon }=y^{A}\underset{0}{\overset{1}{\int }}%
\varepsilon _{A}(x^{i},ty^{B},...,ty_{~i_{1}...i_{r}}^{B})dt
\label{VT_Lagrangian}
\end{equation}%
of $\varepsilon $ and $\varepsilon $ itself:%
\begin{equation}
\kappa (\varepsilon )=E(\lambda _{\varepsilon })-\varepsilon .
\label{canonical_completion}
\end{equation}%
The source form $\kappa (\varepsilon )$ can be expressed, \cite%
{canonical-var-compl}, in terms of the coefficients of the Helmholtz form of 
$\varepsilon ,$ which give a measure of the non-variationality of $%
\varepsilon .$

\bigskip

Assume, in the following, that the background manifold is $Y^{(b)}=Met(X),$
with fibered coordinates $(x^{i},g_{ij});$ in this case, the energy-momentum
source form $\tau =\tau ^{ij}\omega _{ij}\wedge \omega _{0}$ and the
energy-momentum tensor density $\mathcal{T}$ are related by: 
\begin{equation*}
\mathcal{T}^{ij}=-2\tau ^{ij}
\end{equation*}%
(the minus sign appears because, in the relation $\mathcal{T}%
_{~i}^{j}=C_{(hl)i}^{j}\tau ^{hl},$ $C_{(hl)i}^{j}=-\delta
_{h}^{j}g_{li}-\delta _{l}^{j}g_{hi}$).

If we know a term $\varepsilon $ of the energy-momentum source form $\tau ,$
a Lagrangian and, accordingly, the full expression of $\tau $ can be found
using (\ref{VT_Lagrangian}), (\ref{canonical_completion}).

\bigskip

\textbf{Example}: \textit{energy-momentum tensor of a perfect fluid in
general relativity}. Assume that $\dim X=4,$ the signature of the metric is $%
(+,-,-,-)$ and physical measurement units are such that $c=1.$ For an ideal
fluid, the full expression of the energy-momentum tensor is%
\begin{equation}
T^{ij}=(p+\rho )u^{i}u^{j}-pg^{ij},  \label{em_tensor_fluid}
\end{equation}%
where $\rho $ denotes the density and $p,$ the isotropic pressure of the
fluid (both defined in a local inertial frame in which the fluid is at rest, 
\cite{di Felice}), $u^{i}=\dfrac{dx^{i}}{d\tau }$ are the components of the
4-velocity of the fluid and $d\tau =\sqrt{g_{ij}dx^{i}dx^{j}}.$

\bigskip

Let us take, for instance, the term $\rho u^{i}u^{j}$ and find its canonical 
$Met(X)$-variational completion. Build the source form $\varepsilon
=\varepsilon ^{ij}\omega _{ij}\wedge \omega _{0}$ on $Met(X),$ with
coefficients $\varepsilon ^{ij}=\varepsilon ^{ij}(x^{k},g_{kl})$ given by:%
\begin{equation}
\varepsilon ^{ij}=\alpha \rho u^{i}u^{j}\sqrt{\det (g_{ij})},~\ \ \alpha \in 
\mathbb{R}  \label{epsilon}
\end{equation}%
and let us study the behavior of $\varepsilon ^{ij}$ with respect to
homotheties $\chi _{t}:g_{ij}\mapsto tg_{ij},$

The density $\rho $ is inverse proportional to the (3-dimensional) spatial
volume, given, in a local inertial coframe $\{e^{i}\}$, by the volume
element $dV_{3}=e^{1}\wedge e^{2}\wedge e^{3};$ passing to an arbitrary
frame, this is, \cite{Landau}, $dV_{3}=\dfrac{\sqrt{\left\vert \det
(g_{kl})\right\vert }}{\sqrt{g_{00}}}dx^{1}\wedge dx^{2}\wedge dx^{3}.$
Using the relations: $\sqrt{\left\vert \det (g_{ij})\right\vert }\circ \chi
_{t}=t^{2}\sqrt{\left\vert \det (g_{ij})\right\vert },$ $\sqrt{g_{00}}\circ
\chi _{t}=t^{1/2}\sqrt{g_{00}}$, we find: $\rho \circ \chi _{t}=t^{-3/2}\rho
.$ Using the expression $u^{i}=\dfrac{dx^{i}}{d\tau },$ we have $u^{i}\circ
\chi _{t}=t^{-1/2}u^{i}.$ Substituting into (\ref{epsilon}), 
\begin{equation*}
\varepsilon ^{ij}\circ \chi _{t}=t^{-1/2}\alpha \rho u^{i}u^{j}\sqrt{%
\left\vert \det (g_{ij})\right\vert };
\end{equation*}%
Since $g_{ij}u^{i}u^{j}=1,$ we get the Vainberg-Tonti Lagrangian $\lambda
_{\varepsilon }=\mathcal{L}_{\varepsilon }\omega _{0},$ with:%
\begin{equation*}
\mathcal{L}_{\varepsilon }=g_{ij}\underset{0}{\overset{1}{\int }}\varepsilon
^{ij}\circ \chi _{t}dt=\alpha \rho \sqrt{\left\vert \det (g_{ij})\right\vert 
}\underset{0}{\overset{1}{\int }}t^{-1/2}dt=2\alpha \rho \sqrt{\left\vert
\det (g_{ij})\right\vert }.
\end{equation*}%
Variation of $\lambda _{\varepsilon }$ with respect to the metric (as in 
\cite{di Felice}, p. 197) gives then: $T^{ij}=-2\alpha \{(p+\rho
)u^{i}u^{j}-pg^{ij}\}.$ The term $\rho u^{i}u^{j}$ is obtained for $\alpha
=-1/2;$ for this value, the Lagrangian is nothing but the known Lagrangian: 
\begin{equation}
\lambda _{fluid}=-\rho \sqrt{\left\vert \det (g_{ij})\right\vert }\omega
_{0}.  \label{fluid_L}
\end{equation}

The energy-momentum tensor (\ref{em_tensor_fluid}) and the Lagrangian (\ref%
{fluid_L}) were found in \cite{inv-problem-book} starting from the term $%
pg^{ij}.$

\end{document}